\documentclass[aap,preprint]{imsart}%,noinfoline]{imsart}

\usepackage{amsfonts,psfrag,epsfig}
\usepackage{color}
\usepackage{amsmath,amssymb}
\usepackage[amsmath,thmmarks]{ntheorem}
\newtheorem{theorem}{Theorem}
\newtheorem{lemma}[theorem]{Lemma}

\newtheorem{proposition}[theorem]{Proposition}

\newtheorem{definition}{Definition}

{
\theoremstyle{nonumberplain}
\theoremsymbol{\ensuremath{\Box}}
\theoremheaderfont{\normalfont\itshape}
\theorembodyfont{\normalfont}
\newtheorem{proof}{Proof.}

\theoremstyle{empty}

}

\newcommand{\bind}{\bold{I}}

\newcommand{\bW}{\boldsymbol{W}}

\newcommand{\sig}{\boldsymbol{\sigma}}
\newcommand{\brho}{\boldsymbol{\rho}}
\newcommand{\bpi}{\boldsymbol{\pi}}
\newcommand{\bpit}{\widehat{\boldsymbol{\pi}}}
\newcommand{\pit}{\widehat{\pi}}
\newcommand{\bsigma}{\boldsymbol{\sigma}}

\newcommand{\rhO}{\boldsymbol{\rho}}
\newcommand{\lamb}{\boldsymbol{\lambda}}
\newcommand{\Lamb}{\boldsymbol{\Lambda}}

\newcommand{\cT}{\mathcal{T}}

\newcommand{\beq}{\begin{eqnarray}}
\newcommand{\eeq}{\end{eqnarray}}
\newcommand{\beqn}{\begin{equation}}
\newcommand{\eeqn}{\end{equation}}
\newcommand{\R}{\mathbb{R}}
\newcommand{\N}{\mathbb{N}}

\newcommand{\Rp}{\mathbb{R}_+}

\newcommand{\bzero}{\mathbf{0}}
\newcommand{\bone}{\mathbf{1}}

\newcommand{\mc}{\mathcal}
\newcommand{\mb}{\mathbf}

\newcommand{\cM}{\mc{M}}

\newcommand{\cX}{\mc{X}}

\newcommand{\cE}{\mc{E}}
\newcommand{\cG}{\mc{G}}

\newcommand{\cI}{\mc{I}}

\newcommand{\bu}{\mb{u}}
\newcommand{\bv}{\mb{v}}

\newcommand{\bx}{\mb{x}}

\newcommand{\bQ}{\mathbf{Q}}

\newcommand{\E}{\mathbb{E}}

\newcommand{\y}{\boldsymbol{y}}

\DeclareMathOperator{\Conv}{\textsf{Conv}}

\begin{document}
\begin{frontmatter}

\title{Efficient Queue-based CSMA with Collisions}

\runtitle{Randomized Network Scheduling}

\begin{aug}
  \author{\fnms{D.} \snm{Shah} \quad \fnms{J.} \snm{Shin}\thanksref{shin}\ead[label=e1]{jinwoos@mit.edu}}
  \runauthor{Shah \& Shin}
  \affiliation{Massachusetts Institute of Technology}
\thankstext{shin}{Both authors are with the Laboratory for
Information and Decision Systems at MIT. DS and JS are with
the department of EECS and Mathematics, respectively.
Authors' email addresses: {\tt \{devavrat, jinwoos\}@mit.edu}}
\end{aug}

\vspace{.1in}

\begin{abstract}

Recently there has been considerable interest in the
design of efficient carrier sense multiple access(CSMA)
protocol for wireless network starting works by
\cite{RS08}\cite{RSS09}\cite{JW08}\cite{SS}\cite{JSSW}.
The basic assumption underlying these results is availability
of perfect carrier sense information. This allows for
design of continuous time algorithm under which collisions
are avoided.

The primary purpose of this note is to show how these
results can be extended in the case when carrier sense
information may not be perfect, or equivalently delayed.
Specifically, an adaptation of algorithm in \cite{RSS09, SS}
is presented here for time slotted setup with carrier
sense information available only at the end of the time slot.
To establish its throughput optimality, in additon to
method developed in \cite{RSS09, SS}, understanding
properties of stationary distribution of a certain {\em non-reversible}
Markov chain as well as bound on its mixing time is essential.
This note presents these key results.

A longer version of this note will provide detailed account
of how this gets incorporated with methods of \cite{RSS09, SS}
to provide positive recurrence of underlying network Markov
process. In addition, these results will help design optimal
rate control in conjuction with CSMA in presence of
collision building upon method of \cite{JSSW}.

%Our
%arguments for establishing stability (positive Harris
%recurrence) is based on identifying an appropriate
%Lyapunov

\end{abstract}



\end{frontmatter}

\section{Setup}\label{sec:model}

We consider a {\em single-hop} wireless network of $n$ queues.
Queues receive work as per exogeneous arrivals and work leaves
the system upon receiving service. Time is slotted and indexed
by $\tau \in \{0,1,\dots\}$. Arrival process is assumed to be
discrete time and brings unit sized packets. Let
$Q_i(\tau) \in \N$ be number of packets waiting at the $i$th
queue in the begining of time slot $\tau$. Let $A_i(\tau)$
be the total number of packets arrived to queue $i$ till
the end of time slot $\tau$. For convenience, we shall
assume that in a given time slot, arrivals happen at the
end of the time slot. Also assume $A_i(\cdot)$ is a Bernoulli
i.i.d. process with rate $\lambda_i$, i.e.
$\lambda_i = \Pr(A_i(\tau)-A_i(\tau-1) = 1)$ and
$A_i(\tau)-A_i(\tau-1) \in \{0,1\}$ for all $i, \tau \geq 1$.
Let $\bQ(\tau)=[Q_i(\tau)]_{1{\le}i{\le}n}$ and initially
$\tau=0$, $\bQ(0) = \bzero$\footnote{Bold letters are reserved
for vectors; $\bzero, \bone$ represent vectors of all $0$s \& all $1$s
respectively.}.

The work from queues is served at the unit rate, but subject to
{\em interference} constraints. Specifically, let $G = (V,E)$
denote the inference graph between the $n$ queues, represented by
vertices $V = \{1,\dots n\}$ and edges $E$: an $(i,j) \in E$ implies
that queues $i$ and $j$ can not transmit simultaneously since
their transmission {\em interfere} with each other. Formally,
let $\sigma_i(\tau) \in \{0,1\}$ denotes whether the queue $i$
is transmitting at time $\tau$, i.e. work in queue $i$ is being served
at unit rate at time $\tau$ and $\sig(\tau) = [\sigma_i(\tau)]$.
Then, it must be that for $\tau \in \N$,
$$\sig(\tau) \in \cI(G) \stackrel{\Delta}{=}
\{ \rhO = [\rho_i] \in \{0,1\}^n :
   \rho_i + \rho_j \le 1\text{ for all }(i,j) \in E \}.$$
We shall assume that if a non-empty queue $i$ is served in time
slot $\tau$, i.e. $Q_i(\tau) \geq 1$ and $\sigma_i(\tau) = 1$
then a packet departs from it near the end of the time slot
$\tau$, but before arrival happens. In summary, queueing dynamics:
for any $\tau \geq 0$ and $1\leq i\leq n$,
$$ Q_i(\tau+1) = Q_i(\tau) - \sigma_i(\tau) \bind_{\{Q_i(\tau) > 0\}} + A_i(\tau). $$

%To serve packets in queues, each queue can request the network for the availability of transmission i.e.
%``listen'' to the medium whether its interfering neighbors are transmitting or not.
%Depending on the response of the network, the queue $i$ can start transmitting (or served) i.e. become active.
%The total amount of packets (or work) served at the queue $i$ is decided by the total amount of time when
%$i$ is active.

\subsection{Scheduling constraints}  The scheduling algorithm decides
the schedule $\bsigma(\tau) \in \cI(G)$ in the begining of each time
slot, possibly using $\bQ(\tau)$ and past history. This decision is
made in a distributed manner by nodes. Specifically, in the beginning
of each time slot, each node makes a decision to transmit or not.
At the end of the time slot, node knows the following:
\begin{itemize}
\item[$\circ$] if it attempted to transmit, whether its attempt
was successful;
\item[$\circ$] if it did not attempt to transmit, whether any of its
neighbor attempting to transmit was successful.
\end{itemize}
In summary, each node has delayed carrier sense information
that is available at the end of the time slot.

\subsection{Capacity region} From the perspective of network
performance, we would like the scheduling algorithm to be such
that the queues in network remain as small as possible for the
largest possible range of arrival rate vectors.  To formalize
this notion of performance, we define the capacity region.
Let $\Lamb$ be the capacity region defined as
\begin{eqnarray}
 \Lamb & = & \Conv(\cI(G))\nonumber\\
 &=&\left\{ \y \in \Rp^n : \y \leq \sum_{\sig \in \cI(G)} \alpha_{\sig} \sig,
~\mbox{with}~\alpha_{\sig} \geq 0,~\mbox{and}~\sum_{\sig \in \cI(G)} \alpha_{\sig} \leq 1 \right\}.
\label{eq:c1}
\end{eqnarray}
\begin{definition}[throughput optimal] {\em A scheduling algorithm is
called \\ { throughput optimal}, or {stable}, or
providing {100\% throughput}, if for  any $\lamb \in \Lamb^o$ the
(appropriately defined) underlying network Markov process is
\emph{positive (Harris) recurrent}}.
\end{definition}

\section{Our algorithm}\label{sec:algo}

We present a randomized  algorithm that is direct adaptation of the
algorithm in \cite{RSS09, SS} for the discrete time setting.

In the beginning of each time slot, say $\tau$, each node (or queue)
does the following. With probability $1/2$, independent of everything
else, it does nothing. Otherwise, it executes the following:
\begin{itemize}

\item[1.] If $\sigma_i(\tau-1) = 1$, that is its transmission
at time $\tau-1$ was successful, then it decides to transmit with
probability $1-\frac{1}{W_i(\tau)}$.

\item[2.] If at time $\tau-1$, any of its neighbor's transmission
was successful, then does not attempt to transmit with probability $1$.

\item[3.] Otherwise, it attempts transmission with probability $1$.

\end{itemize}

%\textcolor{red}
{Few remarks about the algorithm. In case 1, we choose
\begin{align}\label{eq:func}
   W_i(\tau) & = \exp\left(\max\left\{f(Q_i(\tau)), \sqrt{f(Q_{\max}(\tau))}\right\}\right),
\end{align}
where $f: \R_+ \to [0,\infty)$ is a strictly increasing function with $f(0) = 0$,
$\lim_{x\rightarrow\infty} f(x)=\infty$ and satisfies the property
$$\lim_{x\rightarrow\infty} \exp(f(x))\cdot
f^{\prime}\left(f^{-1}\left(\delta \, f(x)\right)\right)=0,\qquad \text{for any}~~\delta\in(0,1).$$
For example, any strictly increasing function with $f(0)=0$ and
$f(x) = o(\log x)$ will have this property, e.g. $f(x) = \sqrt{\log (x+1)},~
\log \log (x+e),$ etc.}

In above $Q_{\max}(\cdot) = \max_{i} Q_i(\cdot)$, that is the
maximum of all queue sizes. Of course, knowing this instantly
is not possible. However, knowledge of $Q_{\max}(\cdot) \pm O(1)$
suffices and a simple scheme to achieve this is presented in
\cite{RSS09}. Of course, authors strongly believe that explicit
information exchange for knowing such an estimate is needed.

Finally, it is assumed that if a
node tries to attempt as part of the above algorithm, then
it must send some data irrespective of the value of $Q_i(t)$.

\section{Properties of algorithm}

To establish throughput optimality of the algorithm described
above building upon method of \cite{RSS09, SS} will require
us to understand property of the stationary distribution of
a certain Markov chain of the space of independent sets $\cI(G)$
as well as its mixing time. We study these two properties here.
Relation of this Markov chain to algorithm of Section \ref{sec:algo}
is explained.

As mentioned earlier, a longer version of this note will
provide detailed proof of throughput optimality using these
properties.

\subsection{A Markov chain \& its mixing time}\label{ssec:MC}

Consider a graph $G = (V,E)$ of $n = |V|$ nodes
with node weights $\bW = [W_i] \in \R_{\geq 1}^n$
where $\R_{\geq 1} = \{x \in \R : x \geq 1\}$. We consider
Markov chain on the space of independent sets
of $G$, $\cI(G)$ based on $\bW$ with certain
qualitative properties. In what follows we define
what are feasible transitions as part of the chain
and provide properties of the corresponding
transition probabilities. This may not lead to
an exact definition of the Markov chain, i.e. a
class of Markov chains can satisfy these properties.
However, as we shall show that all Markov chains
with these properties have desired properties
in terms of stationary distribution and their
mixing times.

Now we describe what sorts of transitions are
allowed and properties of the corresponding transition
probabilities. Suppose the Markov chain is currently
in the state $\sig \in \cI(G)$. With abuse of notation,
let $\sig$ denote the subset of $V$ that
$\{i \in V: ~\sigma_i = 1\}$.  Then, under the Markov
chain of interest, transition from $\sig$ to $\sig'$
is allowed if and only if $\sig' = \sig \cup S_2 \backslash S_1$
where  $S_1 \subset \sig$ and $S_2 \subset V$ such that
$\sig \cup S_2 \in \cI(G)$. The probability of this
transition, say $P_{\sig\sig'}$ is such that
\[
P_{\sig\sig'} \propto \left(\prod_{i \in S_1}\frac{1}{W_i}\right) p(S_2),
\]
where $2^{-n} \leq p(S_2) \leq 1$. Let $P = [P_{\sig\sig'}] \in [0,1]^{|\cI(G)\times \cI(G)|}$
denote the transition probability matrix.

Under this Markov chain, there is strictly positive probability to
reach empty set, $\bzero$, from any other state $\sig \in \cI(G)$ and
vice versa; empty set has a self loop. Therefore, the Markov chain
is irreducible, aperiodic. It is finite state and hence it has unique
stationary distribution, say $\bpi$. We claim the following
two properties of the Markov chain $P$: first is about $\bpi$
and the second is about its mixing time.
\begin{lemma}\label{lem:goodpi}
For any $\bW \in \R_{\geq 1}^n$,
\[
\E_{\bpi}\left[ \sum_i \sigma_i \log W_i \right] \geq \left(\max_{\brho \in \cI(G)} \sum_i \rho_i \log W_i\right) - O\left({n2^n}\right).
\]
%where the $O(1)$ term depends on $n$ only.
\end{lemma}
\begin{proof}
To start with, it is clear that the stationary distribution
$\bpi$ of the Markov chain $P$ has $\cI(G)$ as its support.
That is, $\bpi = [\pi_{\sig}]_{\sig \in\cI(G)}$ with
$\pi_{\sig} > 0$ for all $\sig \in \cI(G)$. Therefore,
we can write
\begin{align}\label{eq:p1}
\pi_{\sig} & \propto \exp\left( U(\sig) \right),
\end{align}
for some $U : \cI(G) \to \R_+$ where $\R_+ = \{ x \in \R : x\geq 0\}$.
We will show that for all $\sig \in \cI(G)$
\begin{align}\label{eq:p2}
\Bigl| U(\sig) - \sum_i \sigma_i \log W_i \Bigr| & = O\Bigl(n2^n\Bigr).
\end{align}
Assuming \eqref{eq:p2}, we shall conclude the result
of Lemma \ref{lem:goodpi}. For this, we wish to utilize
the following proposition that is a direct adaptation
of the known results in literature (cf. \cite{GBook}
or see \cite{SS}).
\begin{proposition}\label{prop:goodpi}
Let $T: \Omega \to \R$ and let $\cM(\Omega)$ be space of all
distributions on $\Omega$. Define $F : \cM(\Omega) \to \R$ as
    $$F(\mu) = \E_{\mu}(T(\bx)) + H_{ER}(\mu),$$
    where $H_{ER}(\mu)$ is the standard discrete entropy of $\mu$.
    Then, $F$ is uniquely maximized by the distribution $\nu$, where
    $$ \nu_\bx \propto \exp\left(T(\bx)\right),~~\mbox{for any}~~\bx \in \Omega.$$
    Further, with respect to $\nu$, we have
    $$ \E_{\nu}[T(\bx)] \geq \left[\max_{\bx \in \cX} T(\bx)\right] - \log |\Omega|. $$
    \end{proposition}
Now by applying Proposition \ref{prop:goodpi} with $\nu$ replaced by
$\bpi$, $\Omega$ replaced $\cI(G)$ and $T$ replaced by $F$, we have
that
\begin{align}\label{eq:p3}
\E_{\bpi}\left[F(\sig)\right] & \geq \left[\max_{\brho \in \cI(G)} F(\sig)\right] - \log |\cI(G)| \nonumber \\
& \geq \left[\max_{\brho \in \cI(G)} F(\sig)\right] - n,
\end{align}
since $|\cI(G)| \leq 2^n$. Using \eqref{eq:p2} and \eqref{eq:p3}, it follows
that
\begin{align}\label{eq:p4}
\E_{\bpi}\left[\sum_i \sigma_i\log W_i\right] &
\geq \left[\max_{\brho \in \cI(G)} \sum_i \rho_i \log W_i \right] - O\Bigl(n2^n\Bigr).
\end{align}
To complete the proof of Lemma \ref{lem:goodpi}, we shall establish
the remaining claim \eqref{eq:p2}. To this end, consider a different Markov
chain on $\cI(G)$ with transition probability matrix $Q = [Q_{\sig\sig'}]$
such that $Q_{\sig\sig'} > 0$ if and only if $P_{\sig\sig'} > 0$. Now if
$P_{\sig\sig'} > 0$, then it must be that there are $S_1 \subset \sig$,
$S_2 \subset V$ so that $\sig \cup S_2 \in \cI(G)$ and in this case,
we define $Q_{\sig\sig'}$ as
\begin{align}\label{eq:p5}
Q_{\sig\sig'} & \propto \frac{1}{2^n} \prod_{i \in S_1} \frac{1}{W_i}.
\end{align}
Thus, we have that for all $\sig,\sig' \in \cI(G)$ with
$P_{\sig\sig'}, Q_{\sig\sig'} > 0$,
\begin{align}\label{eq:p6}
2^{-n} & \leq \frac{P_{\sig\sig'}}{Q_{\sig\sig'}} ~\leq~ 1.
\end{align}
It can be checked that $Q$, like $P$, is irreducible and
aperiodic Markov chain on $\cI(G)$. Let $\bpit$ be the unique
stationary of $Q$ on $\cI(G)$. We claim that
\begin{align}\label{eq:p7}
\pit_{\sig} & \propto \prod_{i: \sigma_i = 1} {W_i} ~=~ \exp\left(\sum_i \sigma_i \log W_i\right).
\end{align}
To establish this, note that if transition from $\sig$ to $\sig'$
is feasible under $Q$ (equivalently under $P$) then so is
from $\sig'$ to $\sig$. Specifically, let $\sig = S_0 \cup S_1$
and $\sig' = S_0 \cup S_2$, where $S_0, S_1, S_2$ are
disjoint sets and $S_0 \cup S_1 \cup S_2$ ($=\sig \cup S_2$)
is an independent set of $G$. Then,
\begin{align}\label{eq:p8}
\pit_{\sig} Q_{\sig\sig'} & = \left(\prod_{i: \sigma_i = 1} W_i\right) \times \left(\prod_{k \in S_1} \frac{1}{W_k}\right) \times 2^{-n} \nonumber \\
& = \left(\prod_{i \in S_0 \cup S_1} W_i\right) \times \left(\prod_{k \in S_1} \frac{1}{W_k}\right) \nonumber \\
& =  \left(\prod_{i\in S_0} W_i\right) \times 2^{-n} \nonumber \\
& = \left(\prod_{i\in S_0 \cup S_2} W_i\right) \times \left(\prod_{k \in S_2} \frac{1}{W_k}\right) \times 2^{-n} \nonumber \\
& = \pit_{\sig'} Q_{\sig'\sig}.
\end{align}
The \eqref{eq:p8} establishes that $Q$ is reversible
and satisfies detailed balance equation with $\bpit$ as its
stationary distribution. This establishes \eqref{eq:p7}.

Given \eqref{eq:p7}, to establish \eqref{eq:p2} as desired,
it is sufficient to show that for any $\sig\in\cI(G)$,
\begin{align}\label{eq:p9}
2^{-n2^n} & \leq \frac{\pi_{\sig}}{\pit_{\sig}} ~\leq 2^{n2^n}.
\end{align}
To establish this, we shall use the characterization of stationary
distributions for any irreducible, aperiodic finite state
Markov chain given through what is known as the `Markov chain
tree theorem' (cf. see \cite{AT89}). To this end,
define a directed graph $\cG = (\cI(G), \cE)$ with $\cI(G)$ as
vertices and directed edge $(\sig,\sig') \in \cE$ if and only
if $P_{\sig,\sig'} > 0$ (equivalently $Q_{\sig,\sig'} > 0$).
Let $\cT_{\sig}$ be the space of all directed spanning trees of
$\cG$ rooted at $\sig \in \cI(G)$. Define weight of a tree
$T \in \cT_{\sig}$ with respect to transition matrix $P$,
denoted as $w(T,P)$, as
\[
w(T,P) = \prod_{(\brho,\brho') \in T} P_{\brho,\brho'}.
\]
Similarly, define weight of $T \in \cT_{\sig}$ with respect to $Q$,
denoted as $w(T,Q)$, as
\[
w(T,Q) = \prod_{(\brho,\brho') \in T} Q_{\brho,\brho'}.
\]
Then, the Markov Tree Theorem states that for any $\sig \in \cI(G)$,
\begin{align}\label{eq:p10}
\pi_{\sig} & \propto \sum_{T \in \cT_{\sig}} w(T,P).
\end{align}
And, similarly for $\sig \in \cI(G)$,
\begin{align}\label{eq:p11}
\pit_{\sig} & \propto \sum_{T \in \cT_{\sig}} w(T,Q).
\end{align}
Since the number of edges in each spanning tree is no more
than $|\cI(G)| \leq 2^n$, by \eqref{eq:p6}, \eqref{eq:p10}
and \eqref{eq:p11}, it follows that for all $\sig\in \cI(G)$
\begin{align}\label{eq:p12}
2^{-n2^n} & \leq \frac{\pi_{\sig}}{\pit_{\sig}} ~\leq 2^{n2^n}.
\end{align}
This completes the proof of \ref{eq:p9} and subsequently
that of Lemma \ref{lem:goodpi}.
\end{proof}

%\textcolor{red}
{Now we will obtain a mixing rate (or time) of the non-reversible Markov chain $P$.
To this end, we present a bound of the matrix norm of $P^*$ since
it crucially determines the mixing rate of $P$ (cf. \cite{MT06}). Here,
$P^*$ is the adjoint matrix of $P$ and the matrix norm $\|P^*\|$
is defined as
$$\|P^*\| = \sup_{\bv : \E_{\pi}[\bv]=0 } {\frac{{\|P^* \bv\|}_{2,\pi}}{\|\bv\|_{2,\pi}}}, $$
where $\|\bu \|_{2,\pi}=\sqrt{\sum_{\sig\in \mathcal{I}(G)} \pi_{\sig} (u_{\sig})^2}$ 
for any $\bu \in \mathbb{R}^{|\mathcal{I}(G)|}$.}

%\textcolor{red}
{
\begin{lemma}\label{lem:glaumixing}
Given $P$ described above, let $P^*$ be its adjoint. Then,
\beq \|P^*\| &\leq&  1 - \frac1{2^{4n(2^n+2)+2}(W_{\max})^{4n}}.
\eeq
\end{lemma}}
\begin{proof}
We shall use Cheeger's inequality to bound spectral gap for 
reversible Markov chain defined by $PP^*$ and then use it to
bound $\|P^*\|$ using standard result (cf. \cite{MT06}). 

To that end, let $\lambda_2$ and $\lambda_{|\mathcal{I}(G)|}$ be the second-largest and smallest eigenvalues
of $PP^*$, respectively.\footnote{$PP^*$ is reversible, hence all eigenvalues are real  and in the interval $[-1,1]$.}
It is known \cite{MT06} that
\begin{align}
\|P^*\|&=\sqrt{\max\left\{|\lambda_2|, |\lambda_{|\mathcal{I}(G)|}|\right\}}\notag\\
&=\max\left\{\sqrt{1-(1-|\lambda_2|)}, \sqrt{1-(1-|\lambda_{|\mathcal{I}(G)|}|)}\right\}\notag\\
&\leq\max\left\{1-\frac{1-|\lambda_2|}2, 1-\frac{1-|\lambda_{|\mathcal{I}(G)|}|}2\right\}.
%&\leq\max\left\{|\lambda_2|, |\lambda_{|\mathcal{I}(G)|}|\right\}
\label{eq:q4}
\end{align}
First observe that $PP^* \geq \frac1{2^{2n}} I$ (component-wise) since
$P,~P^* \geq \frac1{2^{n}} I$. From this, it is easy to check that
\begin{equation}
\lambda_{|\mathcal{I}(G)|}\geq 2\times \frac1{2^{2n}} - 1.\label{eq:q5}
\end{equation}
Therefore, it suffices to obtain the bound of $\lambda_2$ for the desired bound of $\|P^*\|$
in Lemma \ref{lem:glaumixing}. In general, in the absence of such bound one can
use `lazy' version of the Markov chain, i.e. add self loop to all states with probability $1/2$, 
to make all eigenvalues non-negative and hence need to bound $\lambda_2$ only.
%}

%\textcolor{red}
{
Next, we will use the Cheeger's inequality \cite{C, DFK91, JS, DS, sinclair},
it is well known that $$\lambda_2\leq 1 - \frac{\Phi^2}{2}.$$ Here,  $\Phi$ is the
conductance of $R:=PP^*$, defined as
$$ \Phi ~=~ \min_{S \subset \cI(G)} \frac{Q(S,S^c)}{\min\{\pi(S),\pi(S^c)\}},$$
where $S^c = \cI(G) \backslash S$, $Q(S,S^c) = \sum_{\sig\in S, \sig'\in S^c}{\pi(\sig)R(\sig,\sig')}.$ 
Now we will consider the following naive bounds for $\pi$ and $R$ to derive the desired bound of $\Phi$ and $\lambda_2$.
\begin{align}
\min_{\sig\in \mathcal{I}(G)} \pi_{\sig} &\stackrel{(a)}{\geq} \frac1{2^{n2^n}}\pit_{\sig}\notag\\&\stackrel{(b)}{\geq } \frac1{2^{n2^n}}\times \frac1{2^n
(W_{\max})^n}\notag\\&= \frac1{2^{n(2^n+1)}(W_{\max})^n},\label{eq:q1}
\end{align}}
%\textcolor{red}
{
where (a) and (b) follows from \eqref{eq:p12} and \eqref{eq:p7}, respectively. In addition,
\begin{align}
\min_{R(\sig,\sig')\ne0} R(\sig,\sig') &\geq \min_{P(\sig,\sig')\ne0} P(\sig,\sig')\times \min_{P^*(\sig,\sig')\ne0} P^*(\sig,\sig')\notag\\
&\geq \min_{P(\sig,\sig')\ne0} P(\sig,\sig') \times \left( \min_{\sig\in \mathcal{I}(G)} \pi_{\sig}\times \min_{P(\sig,\sig')\ne0} P(\sig,\sig')\right)\notag\\
&\geq \frac1{2^n} \times \left(\frac1{2^{n(2^n+1)}(W_{\max})^n}\times \frac1{2^n}\right)\notag\\
&= \frac1{2^{n(2^n+3)}(W_{\max})^n}.\label{eq:q2}
\end{align}
Now by the standard application of these bounds \eqref{eq:q1} and \eqref{eq:q2}, we obtain
  \begin{align}
    \Phi &\ge \min_{S\subset \cI(G)}{Q(S,S^c)} \notag\\ & \ge   \min_{R(\sig,\sig')\ne0} \pi_{\sig} R(\sig,\sig') \notag\\
    &\ge \min_{\sig\in \mathcal{I}(G)} \pi_{\sig}\times\min_{R(\sig,\sig')\ne0} R(\sig,\sig')\notag\\ & \geq  \frac1{2^{2n(2^n+2)}(W_{\max})^{2n}}.\label{eq:q3}
  \end{align}
Therefore, from the Cheeger's inequality and \eqref{eq:q3},
\begin{equation}
1-\lambda_2 \geq \Phi^2/2 \geq \frac1{2^{4n(2^n+2)+1}(W_{\max})^{4n}}.\label{eq:q6}
\end{equation}
The desired bound of $\|P^*\|$ follows from \eqref{eq:q4}, \eqref{eq:q5}, \eqref{eq:q6} and
the property $W_{\max}\geq 1$.
This completes the proof of Lemma \ref{lem:glaumixing}.}
 \end{proof}

\subsection{Relation to Algorithm}

Here is a quick explanation of why Markov chain $P$ described
in Section \ref{ssec:MC} arises naturally as part of the algorithm described in
Section \ref{sec:algo}. To that end, the weight vector
$\bW = \bW(\tau)$ is time varying and function of $\bQ(\tau)$
as per \eqref{eq:func}. And, transition of the
set of successfully transmitting nodes $\sig(\tau-1)$ at
time slot $\tau-1$ to the set of successfully transmitting
nodes $\sig(\tau)$ at time $\tau$ is as per the transition
matrix $P=P(\tau)$, where $P=P(\tau)$ has properties described
above with weight $\bW=\bW(\tau)$.

To see this, consider the  $\sig(\tau-1)$. Then, a subset
$S_1 \subset \sig(\tau-1)$ can decide to stop transmitting
at time $\tau$ and these decisions are taken with probability
proportional to $\prod_{i \in S_1} \frac{1}{W_i(\tau)}$.
Clearly, no nodes in neighborhood of $\sig(\tau-1)$ will
attempt transmission as per the algorithm. Therefore, new
nodes attempting transmission must be such that they are not
neighbors of any of the nodes in $\sig(\tau-1)$. For any
subset $S_2$ such that $\sig(\tau-1) \cup S_2 \in \cI(G)$,
it is possible to have $\sig(\tau)$ include $S_2$. This is
because, nodes in $S_2$ attempt transmission and all of
their neighbors (which by definition are not part of $\sig(\tau-1)$)
do not attempt transmission -- this happens with probability
proportional to $2^{-|S_2| - |\Gamma(S_2)|}$, where $\Gamma(S_2)$
are neighbors of $S_2$. Indeed, for $\sig(\tau)$ to  transit
exactly to $\sig(\tau-1) \cup S_2 \backslash S_1$, the overall
probability can be argued in a similar manner to be proportional
to the following:
\[
\left(\prod_{i \in S_1} \frac{1}{W_i(\tau)}\right) \times p'(S_2),
\]
$2^{-n} \leq p'(S_2) \leq 1$. This completes the explanation of
relation between the algorithm and the Markov chain described
in Section \ref{ssec:MC}.

\bibliographystyle{plain}

\bibliography{biblio}

\begin{thebibliography}{10}

\bibitem{AT89}
V.~Anantharam and P.~Tsoucas.
\newblock A proof of the markov chain tree theorem.
\newblock {\em Statistics \& Probability Letters}, 8(2):189--192, June 1989.

\bibitem{C}
J.~Cheeger.
\newblock A lower bound for the smallest eigenvalue of the laplacian.
\newblock {\em Problems in analysis (Papers dedicated to S. Bochner, 1969)},
  pages 195--199, 1970.

\bibitem{DS}
P.~Diaconis and D.~Stroock.
\newblock Geometric bounds on eigenvalues of markov chains.
\newblock {\em Annals of applied probability}, pages 36--61, 1991.

\bibitem{DFK91}
M.~Dyer, A.~Frieze, and R.~Kannan.
\newblock A random polynomial-time algorithm for approximating the volume of
  convex bodies.
\newblock {\em J. ACM}, 38(1):1--17, 1991.

\bibitem{GBook}
H.~O. Georgii.
\newblock {\em Gibbs measures and phase transitions}.
\newblock Walter de Gruyter, 1988.

\bibitem{JS}
M.~Jerrum and A.~Sinclair.
\newblock Polynomial-time approximation algorithms for the ising model.
\newblock {\em SIAM Journal on Computing}, 22:1087--1116, 1993.

\bibitem{JSSW}
L.~Jiang, D.~Shah, J.~Shin, and J.~Walrand.
\newblock Distributed random access algorithm: Scheduling and congesion
  control.
\newblock {\em Available at: arxiv.org/abs/0907.122}.

\bibitem{JW08}
L.~Jiang and J.~Walrand.
\newblock A distributed csma algorithm for throughput and utility maximization
  in wireless networks.
\newblock In {\em Proceedings of 46th Allerton Conference on Communication,
  Control, and Computing, Urbana-Champaign, IL}, 2008.

\bibitem{MT06}
R.~Montenegro and P.~Tetali.
\newblock Mathematical aspects of mixing times in markov chains.
\newblock {\em Found. Trends Theor. Comput. Sci.}, 1(3):237--354, 2006.

\bibitem{RS08}
S.~Rajagopalan and D.~Shah.
\newblock Network scheduling, reversible networks and product-form
  distribution: what do they have in common?
\newblock In {\em CISS, Princeton}, 2008.

\bibitem{RSS09}
S.~Rajagopalan, D.~Shah, and J.~Shin.
\newblock {Network adiabatic theorem: an efficient randomized protocol for
  contention resolution}.
\newblock In {\em Proceedings of the eleventh international joint conference on
  Measurement and modeling of computer systems}, pages 133--144. ACM New York,
  NY, USA, 2009.

\bibitem{SS}
D.~Shah and J.~Shin.
\newblock Randomized scheduling algorithm for queueing networks.
\newblock {\em Under submission, available at Arxiv.org}, 2009.

\bibitem{sinclair}
A.~Sinclair.
\newblock {\em Algorithms for Random Generation and Counting: A Markov Chain
  Approach}.
\newblock Birkh{\"{a}}user, Boston, 1993.

\end{thebibliography}

\end{document}